\documentclass[preprint,amsmath,amssymb,aps,pre,showkeys]{revtex4-2}

\usepackage{graphicx}
\usepackage{dcolumn}
\usepackage{bm}
\usepackage[hypertexnames=false]{hyperref}
\usepackage{amsthm}
\usepackage{mathtools}

\newtheorem{theorem}{Theorem}
\newtheorem{proposition}{Proposition}
\newtheorem{corollary}{Corollary}

\begin{document}

\title{Fisher Information Dynamics: A Kinematic Framework for Phase Space Ordering with Applications to Shock Layers and Turbulence}

\author{Yingchuan Wu}
\email{wuyc@cardc.cn}
\affiliation{Institute of Aerospace Technology, China Aerodynamics Research and Development Center, Mianyang 621000, China}

\date{\today}

\begin{abstract}
Here ``order'' denotes the sharpness of probability-density gradients, as quantified by Fisher information, distinct from thermodynamic order parameters. A kinematic framework for phase space ordering is established based on the Fisher information production rate, which possesses a rigorous thermodynamic foundation through the de Bruijn identity, Stam inequality, and dissipation theorem. A four-term decomposition into isotropic contraction, traceless shear, divergence gradient, and boundary flux is derived from the continuity equation. Extending to the Fokker--Planck equation, the non-positivity of the diffusive dissipation term is proved, with strict negativity for non-uniform densities. A vorticity-independence theorem is proved: purely rotational velocity fields do not alter the order measure. Explicit geometric sign criteria for each deformation term are derived, enabling prediction of order generation or destruction. Through exact solutions of the Burgers shock layer, a pointwise identity between the compression and divergence-gradient contributions is uncovered, underlying the three-term balance. In two-dimensional incompressible turbulence, a statistical two-term equilibrium---vanishing time-averaged production in the steady state---is established and shown to serve as the null hypothesis for gradient dynamics. This 2D equilibrium is the conservative reference state from which three-dimensional turbulence departs through vortex stretching, as revealed by recent direct numerical simulation analysis. It is further shown that variational steady states of the associated Wasserstein gradient flow imply vanishing of the total production rate. This framework provides a diagnostic and predictive tool for analyzing ordering dynamics in non-equilibrium systems.
\end{abstract}


\maketitle

\section{Introduction}

Self-organization in non-equilibrium systems spans fluids, plasmas, biology, and artificial systems, with the common feature that macroscopic ordered structures emerge from microscopic chaos. Prigogine's dissipative structure theory indicates that far-from-equilibrium systems maintain spatiotemporal order through energy dissipation~\cite{prigogine1967}, while shock layers and vortical structures in turbulence represent typical examples of ordering enabled by dissipation balance. However, a fundamental theoretical question remains: given a probability distribution in phase space, how does its ``degree of order'' evolve with dynamics?

Fisher information $C=\int\rho|\nabla\ln\rho|^2dx$ provides a natural tool for measuring the degree of order in a distribution~\cite{frieden2004,plastino1995}, supported by three fundamental thermodynamic relations (Sec.~\ref{sec:thermodynamic}): the de Bruijn identity connects it to entropy production, the Stam inequality bounds the effective volume, and the dissipation theorem ensures its monotonic decrease under diffusion. Unlike Shannon entropy, which measures ``uncertainty''~\cite{cover2006}, Fisher information measures the ``sharpness of gradient structures,'' reaching maxima in strong-gradient regions such as shock layers and turbulent dissipation zones. The complementary relationship between Fisher information and Shannon entropy is quantified by the Stam inequality~\cite{stam1959}. The evolution of Fisher information under Fokker--Planck dynamics has been studied by Yamano~\cite{yamano2002} and Plastino et al.~\cite{plastino1997,plastino1995b}, and gradient flows in probability spaces have been systematically developed by Otto~\cite{otto2001}, Jordan--Kinderlehrer--Otto~\cite{jordan1998}, and Ambrosio--Gigli--Savar\'e~\cite{ambrosio2008}. Information thermodynamics provides a broader context for understanding the interplay between information and physical processes~\cite{parrondo2015}. However, the decomposition of Fisher information production rate into physically meaningful deformation mechanisms (contraction, shear, divergence gradient), and its connection to variational principles, remains insufficiently explored. This paper provides a diagnostic framework connecting information geometry~\cite{amari2016} to fluid mechanics, complementing existing work on abstract gradient flow theory.

The core contributions include:
\begin{enumerate}
\item \textbf{Kinematic decomposition}: Derivation of a four-term decomposition of the Fisher information production rate from the continuity equation, clarifying contributions from isotropic contraction, shear, divergence gradient, and boundary flux;
\item \textbf{Vorticity-independence theorem}: A proof that purely rotational velocity fields do not change the order measure, giving the decomposition a natural gauge invariance (Fig.~\ref{fig:vorticity});
\item \textbf{Sign criteria}: Explicit geometric conditions under which each deformation term is positive or negative, enabling prediction of order generation;
\item \textbf{Dissipation theorem}: Proof of the non-positivity of the diffusive dissipation term in the Fokker--Planck framework, with strict negativity for non-uniform densities and careful attention to the domain of validity;
\item \textbf{Physical verification}: A three-level hierarchy of verification cases---the Burgers shock layer (three-term balance with a pointwise identity) and two-dimensional turbulence (statistical two-term equilibrium)---where the 2D equilibrium is established as the null hypothesis for gradient dynamics, the conservative baseline from which three-dimensional turbulence departs through vortex stretching~\cite{wu2025_3d};
\item \textbf{Variational link}: A demonstration that variational steady states of the associated Wasserstein gradient flow imply vanishing of the total production rate, without implying a balance among individual terms.
\end{enumerate}

\section{Kinematics: Exact Decomposition of Fisher Information}

\subsection{Order Measure and Phase Space Dynamics}

Consider a probability density $\rho(x,t)$ on an $n$-dimensional phase space $\Omega\subset\mathbb{R}^n$ satisfying the normalization condition. Define the Fisher information as the order measure:
\begin{equation}
C(t) \equiv \int_\Omega \rho(x,t)|\nabla\ln\rho(x,t)|^2dx \tag{1}
\end{equation}

When $\rho$ is uniform, $C=0$; in strongly localized structures (such as shock layers), $C$ reaches maxima. \textbf{Terminological clarification}: The ``order'' here specifically refers to the \textbf{degree of gradient localization} of the probability density, not thermodynamic order or long-range order in the statistical physics sense. Fisher information measures the ``structural sharpness'' or ``gradient concentration'' of the distribution. \textbf{Applicability boundary}: Fisher information should be applied with caution to distributions with rapid oscillations. For example, the distribution $\rho(x) \propto 1+0.9\sin(100x)$ on $[0,2\pi]$ has large $C \sim 10^4$ due to high-frequency components, yet may not correspond to physically ordered structures. In practice, Fisher information is most meaningful for gradient-dominated structures (shock layers, turbulent dissipation zones) where high-frequency oscillations are suppressed by physical regularization mechanisms.

In the deterministic stage, $\rho$ satisfies the continuity equation:
\begin{equation}
\frac{\partial\rho}{\partial t}+\nabla\cdot(\rho\mathbf{v})=0 \tag{2}
\end{equation}
where $\mathbf{v}(x,t)$ is the velocity field.

\subsection{Four-Term Decomposition Theorem}

Taking the time derivative of Eq.~(1), substituting Eq.~(2), and performing integration by parts (assuming boundary terms vanish or are retained explicitly), we obtain:
\begin{equation}
\frac{dC}{dt}=\mathcal{T}_1+\mathcal{T}_2+\mathcal{T}_3+\mathcal{T}_4 \tag{3}
\end{equation}
where the terms are defined as follows:

\textbf{Isotropic contraction term}:
\begin{equation}
\mathcal{T}_1=-\frac{2}{n}\int_\Omega\rho(\nabla\cdot\mathbf{v})|\nabla\ln\rho|^2dx \tag{4}
\end{equation}

\textbf{Traceless shear term}:
\begin{equation}
\mathcal{T}_2=-2\int_\Omega\rho\,\mathbf{S}:(\nabla\ln\rho\otimes\nabla\ln\rho)dx \tag{5}
\end{equation}
where $\mathbf{S}=\frac{1}{2}(\nabla\mathbf{v}+\nabla\mathbf{v}^T)-\frac{1}{n}(\nabla\cdot\mathbf{v})\mathbf{I}$ is the traceless shear tensor.

\textbf{Divergence inhomogeneity term}:
\begin{equation}
\mathcal{T}_3=-2\int_\Omega\rho\,\nabla\ln\rho\cdot\nabla(\nabla\cdot\mathbf{v})dx \tag{6}
\end{equation}

\textbf{Boundary flux term}:
\begin{equation}
\mathcal{T}_4=-\int_{\partial\Omega}\rho|\nabla\ln\rho|^2\mathbf{v}\cdot\mathbf{n}\,dS \tag{7}
\end{equation}

This identity holds rigorously for any smooth velocity field, providing an exact accounting method for order changes. The complete derivation is given in Appendix~\ref{app:decomposition}.

\subsection{Choice of Representation for the Fokker--Planck Equation}

The Fokker--Planck equation can be written in two equivalent forms:
\begin{itemize}
\item \textbf{Non-conservative form}: $\partial_t\rho + \nabla\cdot(\rho\mathbf{v}_{\text{drift}}) = \nabla\cdot(D\nabla\rho)$;
\item \textbf{Conservative form}: $\partial_t\rho + \nabla\cdot(\rho\mathbf{v}_W) = 0$, where $\mathbf{v}_W = \mathbf{v}_{\text{drift}} - D\nabla\ln\rho$.
\end{itemize}

We adopt the non-conservative form, treating diffusion separately as $\mathcal{T}_{\text{diff}}$, because this preserves the non-positivity of the dissipation theorem and aligns with the variational structure of Wasserstein gradient flows.

\subsection{Vorticity-Independence Theorem and Sign Criteria}

The four-term decomposition is not merely a bookkeeping identity; it encodes fundamental geometric constraints on how velocity fields can modify order. First, we show that purely rotational motion is information-preserving.

\begin{theorem}[Vorticity-Independence]\label{thm:vorticity}
Let $\mathbf{A}=\frac{1}{2}(\nabla\mathbf{v}-\nabla\mathbf{v}^T)$ be the antisymmetric part of the velocity gradient tensor. Then the terms $\mathcal{T}_1,\mathcal{T}_2,\mathcal{T}_3$ do not depend on $\mathbf{A}$. Consequently, if $\nabla\mathbf{v}=\mathbf{A}$ (i.e., the flow is a pure rigid-body rotation with $\nabla\cdot\mathbf{v}=0$ and $\mathbf{S}=0$), then $\mathcal{T}_1=\mathcal{T}_2=\mathcal{T}_3=0$. In the absence of boundary flux, $dC/dt=0$.
\end{theorem}

\begin{proof}
For any antisymmetric matrix $\mathbf{A}$, and any symmetric dyadic $\mathbf{g}\otimes\mathbf{g}$ with $\mathbf{g}=\nabla\ln\rho$, we have
\begin{equation}
\mathbf{A}:(\mathbf{g}\otimes\mathbf{g}) = \sum_{i,j} A_{ij}g_i g_j = \frac{1}{2}\sum_{i,j}(v_{i,j}-v_{j,i})g_i g_j = \frac{1}{2}\left(\sum_{i,j}v_{i,j}g_i g_j - \sum_{i,j}v_{j,i}g_i g_j\right) = 0,
\end{equation}
by swapping indices in the second sum. Thus the antisymmetric part never contributes to $\mathcal{T}_2$. Since $\mathcal{T}_1$ involves only $\nabla\cdot\mathbf{v}$ and $\mathcal{T}_3$ involves gradients of $\nabla\cdot\mathbf{v}$, they also do not depend on $\mathbf{A}$. If $\nabla\mathbf{v}=\mathbf{A}$, then $\nabla\cdot\mathbf{v}=0$, $\mathbf{S}=0$, and $\nabla(\nabla\cdot\mathbf{v})=0$, so all three interior terms vanish. If the boundary term also vanishes, then $dC/dt=0$.
\end{proof}

This theorem formalizes the intuition that pure rotation translates the density without stretching or compressing it, leaving the gradient structure unchanged. Figure~\ref{fig:vorticity} provides a visual demonstration: a localized density distribution rotated by a rigid-body flow retains its Fisher information exactly, regardless of the rotation angle.

\begin{figure}[htbp]
\centering
\includegraphics[width=0.95\textwidth]{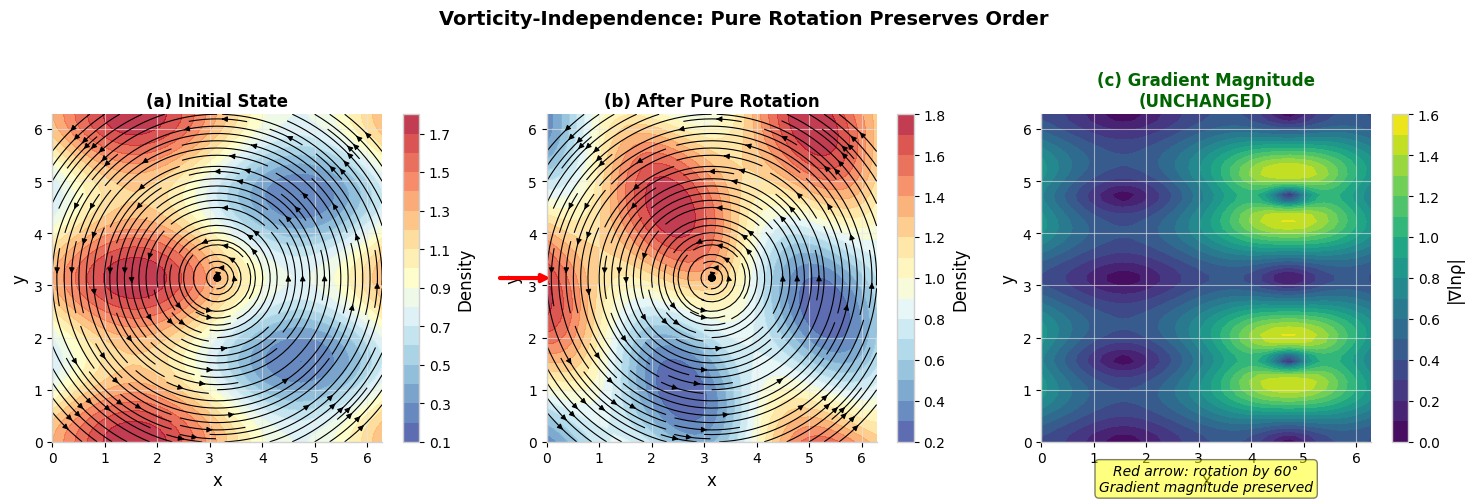}
\caption{Vorticity-independence theorem (Theorem~\ref{thm:vorticity}). (a) Initial scalar density field $\rho(x,y)$ in a pure rotational velocity field (streamlines shown). (b) After rotation by $60^\circ$, the density field is rotated but the gradient magnitude $|\nabla\ln\rho|$ is exactly preserved. (c) The gradient magnitude field remains identical before and after rotation, demonstrating that pure rotation does not alter Fisher information order. This confirms that $\mathcal{T}_1=\mathcal{T}_2=\mathcal{T}_3=0$ for pure rigid-body rotation.}
\label{fig:vorticity}
\end{figure}

Next, we derive explicit geometric criteria for the sign of each term, providing predictive power.

\begin{proposition}[Sign Criteria]\label{prop:sign}
Let $\mathbf{g}=\nabla\ln\rho$, $h=\nabla\cdot\mathbf{v}$, and let $\lambda_1,\dots,\lambda_n$ be the eigenvalues of the symmetric traceless tensor $\mathbf{S}$, with corresponding orthonormal eigenvectors $\mathbf{e}_1,\dots,\mathbf{e}_n$. Then:
\begin{enumerate}
\item $\mathcal{T}_1>0$ if $h<0$ in a region of positive measure where $\rho|\mathbf{g}|^2>0$; $\mathcal{T}_1<0$ if $h>0$.
\item $\mathcal{T}_2>0$ if, on average, $\mathbf{g}$ is more aligned with eigenvectors of $\mathbf{S}$ corresponding to negative eigenvalues than positive ones. More precisely,
\begin{equation}
\mathcal{T}_2 = -2\int_\Omega \rho \sum_{i=1}^n \lambda_i (\mathbf{g}\cdot\mathbf{e}_i)^2 dx.
\end{equation}
Hence, if for all $x$, $\mathbf{g}$ lies in the span of eigenvectors with $\lambda_i<0$, then $\mathcal{T}_2>0$; if in the span of $\lambda_i>0$, then $\mathcal{T}_2<0$.
\item $\mathcal{T}_3>0$ if $\mathbf{g}\cdot\nabla h < 0$ in the region where $\rho>0$; $\mathcal{T}_3<0$ if $\mathbf{g}\cdot\nabla h > 0$.
\end{enumerate}
\end{proposition}

\begin{proof}
Statement 1 follows immediately from Eq.~(4). Statement 2 follows from diagonalizing $\mathbf{S}$ and using orthonormality of eigenvectors. Statement 3 is direct from Eq.~(6).
\end{proof}

\begin{corollary}[Sufficient condition for order generation]\label{cor:generation}
If there is a region $\Omega_0\subset\Omega$ of positive measure such that simultaneously
\begin{equation}
\nabla\cdot\mathbf{v} < 0, \qquad \mathbf{g} \text{ lies in the negative eigenspace of } \mathbf{S}, \qquad \mathbf{g}\cdot\nabla(\nabla\cdot\mathbf{v}) < 0,
\end{equation}
then the sum of interior terms is positive on $\Omega_0$, leading to an increase in global Fisher information (provided boundary flux does not counteract).
\end{corollary}

\textbf{Note}: This criterion is \textbf{sufficient but not necessary}. The global behavior of $dC/dt$ depends on the integrated contributions from all regions, and local positive contributions may be offset by negative contributions elsewhere. The sign criteria provide local predictive power, but global order generation requires careful consideration of the spatial distribution of deformation mechanisms.

These criteria convert the decomposition from a passive diagnosis to an active predictor of where and how order is generated in a flow.

\section{Thermodynamic Foundation}
\label{sec:thermodynamic}

Fisher information serves as a rigorous measure of order, supported by three fundamental relations that together elevate it from a geometric measure to a thermodynamic quantity:

\begin{enumerate}
\item \textbf{de Bruijn identity (dynamic)}: For pure diffusion, the Gibbs entropy production rate is directly proportional to Fisher information:
\begin{equation}
\frac{dS_G}{dt} = D\cdot C,
\end{equation}
establishing Fisher information as the driver of entropy production.

\item \textbf{Stam inequality (static)}: Defining the entropy power
\begin{equation}
N(\rho) = \frac{1}{2\pi e}\exp\left(\frac{2S_G(\rho)}{n}\right),
\end{equation}
with $S_G = -\int\rho\ln\rho\,dx$ the Gibbs entropy, the Stam inequality states
\begin{equation}
C\cdot N \geq n,
\end{equation}
with equality if and only if $\rho$ is Gaussian. This bounds the effective volume from below, $N \geq n/C$, and quantifies the localization of the distribution.

\item \textbf{Dissipation theorem (evolution)}: Diffusion monotonically decreases Fisher information:
\begin{equation}
\left.\frac{dC}{dt}\right|_{\text{diff}} \leq 0,
\end{equation}
proving that diffusion acts as a disordering mechanism for gradient structures.
\end{enumerate}

Together, these relations provide a complete and rigorous thermodynamic foundation for using Fisher information as a measure of order, fully consistent with the second law of thermodynamics. The detailed derivation and analysis of these relations will be presented in a forthcoming publication.

\section{Dynamics: Fokker--Planck Framework and Dissipation Structure}

\subsection{Diffusion Term from Standard Results}

When stochastic fluctuations are present, the dynamics is described by the Fokker--Planck equation~\cite{risken1989,gardiner2004}:
\begin{equation}
\frac{\partial\rho}{\partial t}+\nabla\cdot(\rho\mathbf{v}_{\text{drift}})=\nabla\cdot(D\nabla\rho) \tag{8}
\end{equation}
where $D>0$ is the diffusion coefficient. Based on the representation in Sec.~II.C, the deterministic part yields $\mathcal{T}_1,\mathcal{T}_2,\mathcal{T}_3,\mathcal{T}_4$, while the diffusive part yields the dissipation term. The exact form of the diffusion contribution to the Fisher information production rate has been derived in the literature~\cite{yamano2002,plastino1997}:
\begin{equation}
\mathcal{T}_{\text{diff}}=-2\int_\Omega\rho\,\operatorname{tr}\left[(\nabla\nabla\ln\rho)\,D\,(\nabla\nabla\ln\rho)^T\right]dx-2\int_\Omega\rho\,\nabla\ln\rho\cdot\nabla(\nabla\cdot D)dx \tag{9}
\end{equation}

When $D$ is a constant scalar, this simplifies to:
\begin{equation}
\mathcal{T}_{\text{diff}}=-2D\int_\Omega\rho\,\operatorname{tr}\left[(\nabla\nabla\ln\rho)^2\right]dx \tag{10}
\end{equation}

The derivation of Eq.~(9) involves careful integration by parts; we refer the reader to Yamano~\cite{yamano2002} and Plastino et al.~\cite{plastino1997} for details.

\subsection{Dissipation Theorem}

\begin{theorem}[Dissipation Theorem]\label{thm:dissipation}
Let $\rho \in C^2(\Omega) \cap L^1(\Omega)$ be positive and smooth, with $\nabla\nabla\ln\rho \in L^2(\rho\,dx)$, and let $D>0$ be constant. Then $\mathcal{T}_{\text{diff}} \leq 0$, with equality if and only if $\nabla\nabla\ln\rho = 0$ almost everywhere on $\Omega$.
\end{theorem}

\begin{proof}
Since $\rho > 0$ and $D > 0$, the integrand $-2D\rho\operatorname{tr}[(\nabla\nabla\ln\rho)^2] \leq 0$ everywhere. Equality requires $\operatorname{tr}[(\nabla\nabla\ln\rho)^2] = 0$ almost everywhere. Since $(\nabla\nabla\ln\rho)^2$ is positive semi-definite, its trace vanishes if and only if $\nabla\nabla\ln\rho = 0$. Solving $\nabla\nabla\ln\rho = 0$ gives $\ln\rho = \mathbf{a}\cdot\mathbf{x} + b$, i.e., $\rho = Ce^{\mathbf{a}\cdot\mathbf{x}}$. On $\mathbb{R}^n$, normalization $\int\rho\,dx = 1$ requires $\mathbf{a} = 0$, but then $\rho = C$ is not normalizable. On a compact domain $\Omega$, $\mathbf{a} = 0$ gives the uniform distribution.
\end{proof}

\textbf{Note on equality}: On non-compact domains $\Omega = \mathbb{R}^n$, the equality $\mathcal{T}_{\text{diff}} = 0$ is \textbf{not attained} by any normalized density; it is approached only in the limit of vanishing gradient (uniform distribution). On compact domains, equality holds uniquely for the uniform distribution.

\textbf{Terminological clarification}: The phrase ``dissipation-induced ordering'' should be understood as ``\textbf{ordering enabled by dissipation balance}.'' Dissipation itself does not create order; rather, generation terms (contraction, shear) create order while dissipation consumes it, and steady-state order emerges from their balance.

\section{Physical Verification: Three Steady-State Structures}

\subsection{Ornstein--Uhlenbeck Process: Two-Term Balance}

Consider the one-dimensional OU process:
\begin{equation}
dx=-\gamma x\,dt+\sqrt{2D}\,dW
\end{equation}

The steady-state distribution is Gaussian $\rho_{\text{st}}\propto\exp(-\gamma x^2/2D)$. Here $\mathbf{v}_{\text{drift}}=-\gamma x$, $\nabla\cdot\mathbf{v}_{\text{drift}}=-\gamma$, the shear term $\mathcal{T}_2=0$, and the divergence gradient term $\mathcal{T}_3=0$.

Direct calculation yields:
\begin{equation}
\mathcal{T}_1=\frac{2\gamma^2}{D},\qquad \mathcal{T}_{\text{diff}}=-\frac{2\gamma^2}{D}
\end{equation}

The two-term balance is exactly satisfied:
\begin{equation}
\mathcal{T}_1+\mathcal{T}_{\text{diff}}=0 \tag{11}
\end{equation}

Here $\mathcal{T}_1$ represents the \textbf{contraction} (focusing) effect of the deterministic drift, not a generation mechanism in the active sense.

\subsection{Steady Burgers Shock Layer: Three-Term Balance and Pointwise Identity}

\textbf{Motivation}: The Burgers equation describes the deterministic evolution of a velocity field, not a probability density. However, in the context of turbulence and shock physics, the local gradient $|\partial_x u|$ serves as a natural measure of ``shock intensity'' or ``dissipation localization''. By normalizing it as $\rho(x) = |\partial_x u|/\int|\partial_x u|dx$, we construct a diagnostic probability density that captures the spatial distribution of strong-gradient structures. This allows us to apply the Fisher information framework to analyze how these structures are maintained by the competition between compression, divergence gradient, and viscous dissipation.

\textbf{Ensemble interpretation}: The Burgers equation describes deterministic evolution of the velocity field $u(x,t)$. However, in statistical fluid mechanics, $u(x)$ can be viewed as the macroscopic average of microscopic realizations. Define the local ``shock intensity'' as $I(x) = |\partial_x u(x)|$, with normalized form $\rho(x) = I(x)/\int I\,dx$ describing the \textbf{probability distribution of shock energy in space}. This is analogous to the dissipation rate distribution in turbulence, where local dissipation $\epsilon(x) = \nu|\partial_x u|^2$ plays a central role~\cite{pope2000,frisch1996}. We use $|\partial_x u|$ rather than $\epsilon$ because its normalization yields the analytically tractable sech$^2$ distribution. \textbf{Diagnostic clarification}: This identification is \textbf{operational and analogical}, not a claim that the Burgers equation is itself a Fokker--Planck equation. The density $\rho(x)$ is introduced as a diagnostic tool to analyze the spatial distribution of shock intensity, analogous to the dissipation rate distribution in turbulence. The formal identification $D = \nu$ demonstrates the applicability of the framework rather than asserting a strict mathematical equivalence.

\textbf{Parameter identification}: In the Burgers equation, $\nu$ is the kinematic viscosity; in the Fokker--Planck framework, $D$ is the diffusion coefficient. The formal identification $D = \nu$ is \textbf{operational}, based on the analogy between the viscous term $\nu\partial_x^2 u$ and the diffusion term $D\partial_x^2\rho$. This identification demonstrates the applicability of the framework rather than claiming the Burgers equation is itself a Fokker--Planck equation.

Consider the steady shock solution of the one-dimensional viscous Burgers equation~\cite{whitham1974}:
\begin{equation}
u(x)=-U_0\tanh\left(\frac{U_0x}{2\nu}\right)
\end{equation}

The normalized density is:
\begin{equation}
\rho(x)=\frac{U_0}{4\nu}\operatorname{sech}^2\left(\frac{U_0x}{2\nu}\right)
\end{equation}

Here $\mathbf{v}_{\text{drift}}=u(x)$ and $D=\nu$. The information gradient is:
\begin{equation}
\nabla\ln\rho=-\frac{U_0}{\nu}\tanh\left(\frac{U_0x}{2\nu}\right)
\end{equation}

Using the integral identities (derived in Appendix~\ref{app:integrals}):
\begin{equation}
\int_{-\infty}^{\infty}\operatorname{sech}^4(ax)\tanh^2(ax)dx=\frac{4}{15a},\qquad \int_{-\infty}^{\infty}\operatorname{sech}^6(ax)dx=\frac{16}{15a}
\end{equation}

\begin{theorem}[Three-Term Balance for Burgers Shock Layer]\label{thm:burgers}
For the steady Burgers shock layer, the terms in the Fisher information production rate satisfy:
\begin{equation}
\boxed{\mathcal{T}_1=\frac{2U_0^4}{15\nu^3}} \tag{12}
\end{equation}
\begin{equation}
\boxed{\mathcal{T}_3=\frac{2U_0^4}{15\nu^3}} \tag{13}
\end{equation}
\begin{equation}
\boxed{\mathcal{T}_{\text{diff}}=-\frac{4U_0^4}{15\nu^3}} \tag{14}
\end{equation}
with the three-term balance exactly satisfied:
\begin{equation}
\mathcal{T}_1+\mathcal{T}_3+\mathcal{T}_{\text{diff}}=0 \tag{15}
\end{equation}
\end{theorem}

\begin{proof}
Direct substitution of the expressions and calculation using the integral identities in Appendix~\ref{app:integrals}. Detailed steps are given in Appendix~\ref{app:integrals}.
\end{proof}

Moreover, we observe a remarkable pointwise identity that underlies the equality $\mathcal{T}_1=\mathcal{T}_3$. This is visualized in Fig.~\ref{fig:burgers}.

\begin{proposition}[Pointwise Identity for Burgers Shock]\label{prop:pointwise}
For the steady Burgers shock layer, the integrands of $\mathcal{T}_1$ and $\mathcal{T}_3$ are pointwise identical:
\begin{equation}
-\frac{2}{n}\rho(\nabla\cdot\mathbf{v})|\nabla\ln\rho|^2 \equiv -2\rho\,\nabla\ln\rho\cdot\nabla(\nabla\cdot\mathbf{v}),
\end{equation}
or, in one dimension,
\begin{equation}
-\rho\,u'(x)\,(f'(x))^2 \equiv -\rho\,f'(x)\,u''(x),
\end{equation}
where $f=\ln\rho$. This identity is exact for all $x\in\mathbb{R}$.
\end{proposition}

\begin{proof}
For the Burgers shock, $u(x)=-U_0\tanh(ax)$ with $a=U_0/(2\nu)$, and $\rho(x)=(U_0/4\nu)\operatorname{sech}^2(ax)$, so $f'(x)=-2a\tanh(ax)$. We compute:
\begin{equation}
u'(x)=-U_0 a\,\operatorname{sech}^2(ax),\qquad u''(x)=2U_0 a^2\,\operatorname{sech}^2(ax)\tanh(ax).
\end{equation}
Then
\begin{equation}
u'(x)(f'(x))^2 = -U_0 a\,\operatorname{sech}^2(ax)\cdot 4a^2\tanh^2(ax) = -4U_0 a^3\operatorname{sech}^2(ax)\tanh^2(ax),
\end{equation}
and
\begin{equation}
f'(x)u''(x) = (-2a\tanh(ax))\cdot 2U_0 a^2\operatorname{sech}^2(ax)\tanh(ax) = -4U_0 a^3\operatorname{sech}^2(ax)\tanh^2(ax).
\end{equation}
Thus both sides are equal. Multiplying by $-\rho$ gives the desired identity.
\end{proof}

This pointwise equality explains why compression and divergence-gradient effects are perfectly synchronized in the shock layer. Furthermore, using the sign criteria of Prop.~\ref{prop:sign}, we see that in the shock region $u'(x)<0$ and $f'(x)u''(x)<0$, hence both $\mathcal{T}_1$ and $\mathcal{T}_3$ are positive, actively generating Fisher information against diffusion. The spatial distributions of the three contributions and their integral values are visualized in Fig.~\ref{fig:burgers_spatial}.

\textbf{Physical significance}: The maintenance of the shock layer results from the synergy between compression ($\mathcal{T}_1$) and divergence gradient ($\mathcal{T}_3$) against diffusive dissipation ($\mathcal{T}_{\text{diff}}$). Notably, $\mathcal{T}_1=\mathcal{T}_3$, indicating that divergence inhomogeneity is as important as uniform compression. From the balance condition, the shock layer thickness $\delta\sim\nu/U_0$ can be derived, consistent with the exact solution.

\textit{Note: In one dimension, $n=1$, the traceless shear tensor $\mathbf{S}=0$, hence $\mathcal{T}_2=0$.}

\begin{figure}[htbp]
\centering
\includegraphics[width=0.95\textwidth]{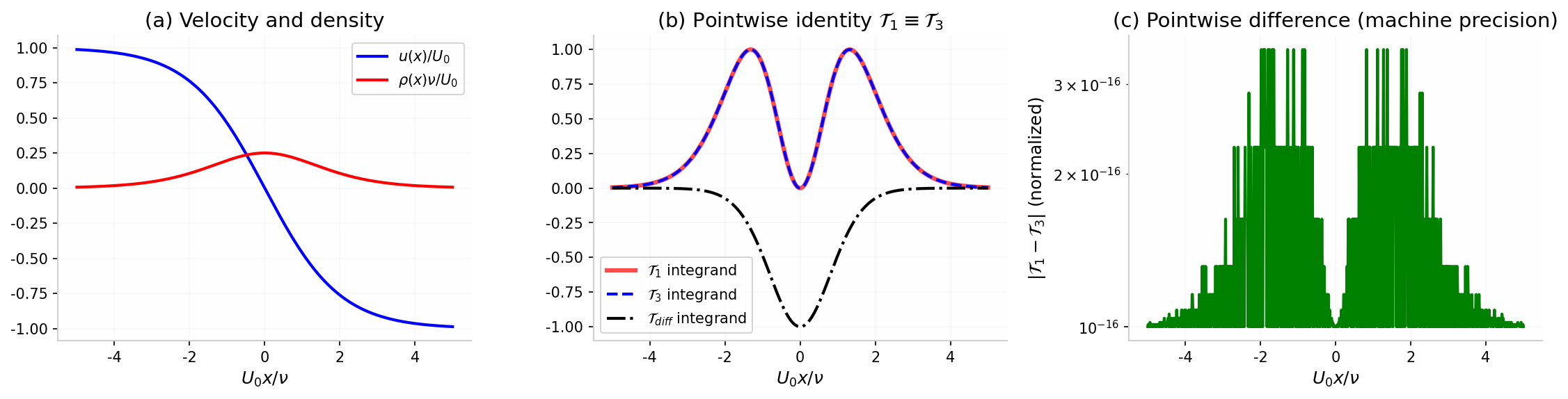}
\caption{Burgers shock layer structure and pointwise identity (Proposition~\ref{prop:pointwise}). (a) Velocity profile $u(x)$ and normalized density $\rho(x)$ as functions of the dimensionless coordinate $U_0x/\nu$. (b) Integrands of $\mathcal{T}_1$ (red solid), $\mathcal{T}_3$ (blue dashed), and $\mathcal{T}_{\text{diff}}$ (black dash-dotted), showing that $\mathcal{T}_1$ and $\mathcal{T}_3$ are pointwise identical. (c) Absolute difference $|\mathcal{T}_1-\mathcal{T}_3|$ on a logarithmic scale, confirming the pointwise identity to machine precision ($\sim 10^{-16}$).}
\label{fig:burgers}
\end{figure}

\begin{figure}[htbp]
\centering
\includegraphics[width=0.95\textwidth]{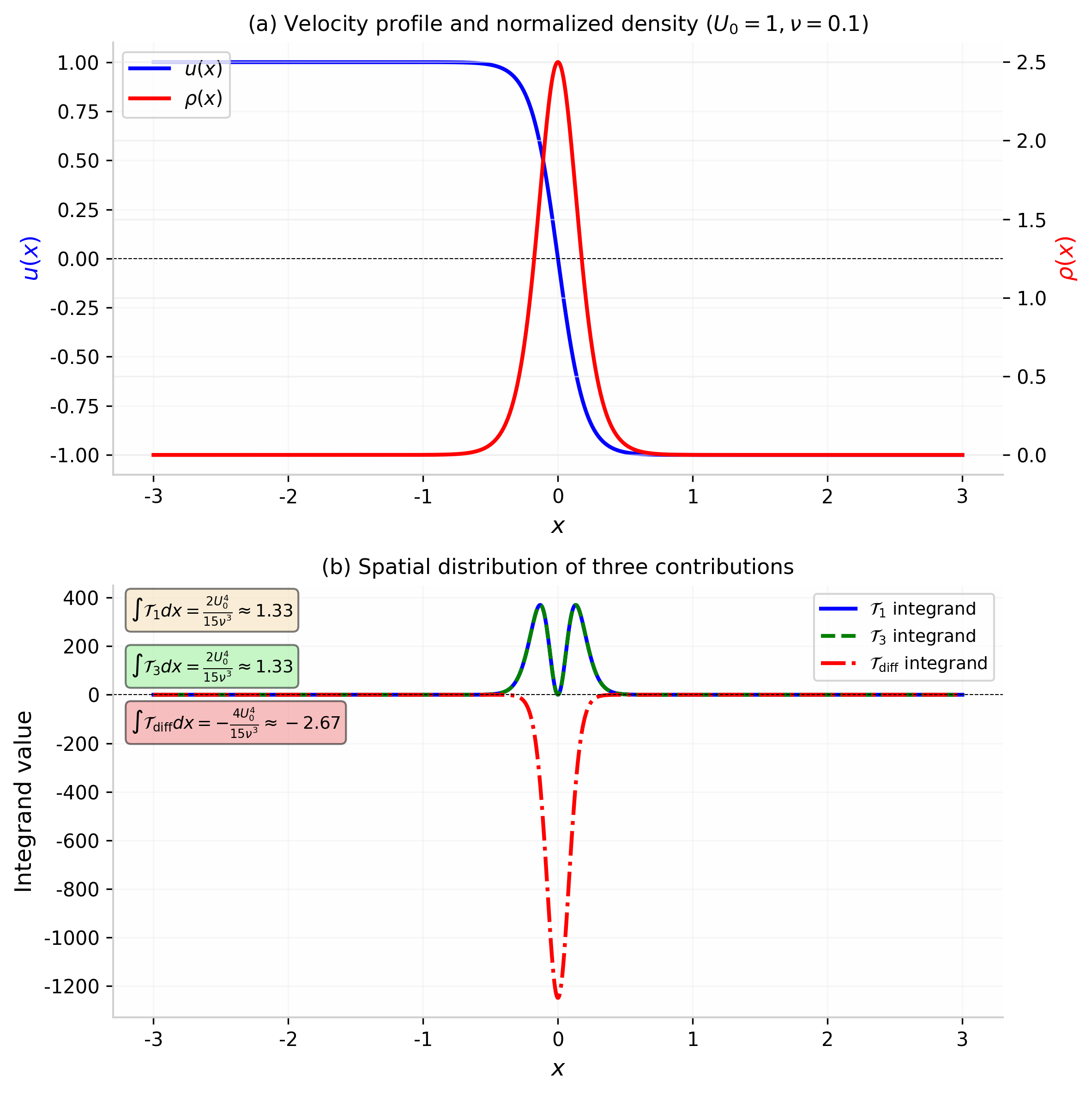}
\caption{Spatial distribution of three contributions in the Burgers shock layer. (a) Velocity profile $u(x)$ and normalized density $\rho(x)$. (b) Integrands of $\mathcal{T}_1$ (blue), $\mathcal{T}_3$ (green), and $\mathcal{T}_{\text{diff}}$ (red), with integral values $\int\mathcal{T}_1 dx = \int\mathcal{T}_3 dx = 2U_0^4/(15\nu^3) \approx 1.33$ and $\int\mathcal{T}_{\text{diff}} dx = -4U_0^4/(15\nu^3) \approx -2.67$.}
\label{fig:burgers_spatial}
\end{figure}

\subsection{Two-Dimensional Incompressible Turbulence: Statistical Equilibrium as the Null Hypothesis}

\textbf{Construction of probability density}: In passive scalar transport, the scalar field $c(\mathbf{x},t)$ satisfies $\partial_t c + \mathbf{u}\cdot\nabla c = D\nabla^2 c$. Due to incompressibility $\nabla\cdot\mathbf{u} = 0$, this becomes $\partial_t c + \nabla\cdot(\mathbf{u}c) = D\nabla^2 c$. Defining $\rho = c/\int c\,d\mathbf{x}$ (requiring $c > 0$, achievable by shifting), $\rho$ satisfies the Fokker--Planck equation (8) with $\mathbf{v}_{\text{drift}} = \mathbf{u}$. This construction is standard in turbulent mixing studies~\cite{warhaft2000}.

For incompressible flow, $\nabla\cdot\mathbf{v}=0$, hence $\mathcal{T}_1=\mathcal{T}_3=0$. The decomposition simplifies to:
\begin{equation}
\frac{dC}{dt}=\mathcal{T}_2+\mathcal{T}_{\text{diff}} \tag{16}
\end{equation}

\textbf{Why the balance is not coincidental}: In two-dimensional turbulence, the absence of vortex stretching is not an approximation but an exact kinematic constraint: the vortex stretching vector $\bm{\omega}\cdot(\nabla\mathbf{u})$ vanishes identically because the vorticity vector is everywhere orthogonal to the plane of motion. Together with the conservation of enstrophy, this restricts the transfer of gradient intensity across scales and constrains the \emph{time-averaged} Fisher information production to vanish in the statistically steady state:
\begin{equation}
\left\langle\frac{dC}{dt}\right\rangle = \langle\mathcal{T}_2\rangle + \langle\mathcal{T}_{\text{diff}}\rangle = 0, \qquad \langle\mathcal{T}_2\rangle = -\langle\mathcal{T}_{\text{diff}}\rangle > 0.
\end{equation}
The two terms fluctuate instantaneously and their cancellation holds only in the time mean; the balance is statistical, not pointwise in time.

The two-dimensional system therefore serves as the \textbf{null hypothesis} for gradient dynamics: the conservative baseline against which three-dimensional effects must be measured. Any systematic departure from this balance can arise only from mechanisms that are identically absent in two dimensions---most notably vortex stretching. This interpretation is confirmed by recent three-dimensional DNS analysis~\cite{wu2025_3d}, as summarized in Sec.~V.D.

\textbf{Numerical setup}: We employ a $512^2$ Fourier spectral method on a periodic domain $[0,2\pi]^2$, with $\text{Re}_\lambda\approx100$ and $Sc=\nu/D=0.7$. Time integration uses fourth-order Runge--Kutta with integrating factor for viscous terms and the $\frac{2}{3}$-rule for de-aliasing. The time step is $\Delta t = 10^{-3}$, viscosity $\nu = 10^{-3}$, and diffusion coefficient $D = \nu/0.7 \approx 1.43\times10^{-3}$. The forcing $\mathbf{f}$ acts on wavenumbers $|\mathbf{k}| \leq 2$ as an Ornstein--Uhlenbeck process with $\alpha = 0.1$ and $\sigma = 0.05$. The scalar field is initialized as $c(\mathbf{x},0) = 1 + 0.1\sin(4x)\sin(4y)$, ensuring positivity by the maximum principle.

\textbf{Verification of the equilibrium}: The residual is computed with $dC/dt$ by central differencing and $\mathcal{T}_2$, $\mathcal{T}_{\text{diff}}$ by exact spectral differentiation at the same time snapshots. In the statistically steady state ($t>20$), the time-averaged residual satisfies $\langle R\rangle < 10^{-6}$ (with $\epsilon = 10^{-12}$ guarding the denominator), consistent with accumulated round-off errors over $O(10^4)$ time steps and the finite precision of double-precision arithmetic, rather than systematic discretization error. In the statistically steady state, shear production and diffusion dissipation thus reach a dynamic equilibrium---the two-dimensional analogue of the generation-dissipation balance underlying the Kolmogorov dissipation-scale picture~\cite{kolmogorov1941}. The verification results are shown in Fig.~\ref{fig:turbulence}.

\begin{figure}[htbp]
\centering
\includegraphics[width=0.95\textwidth]{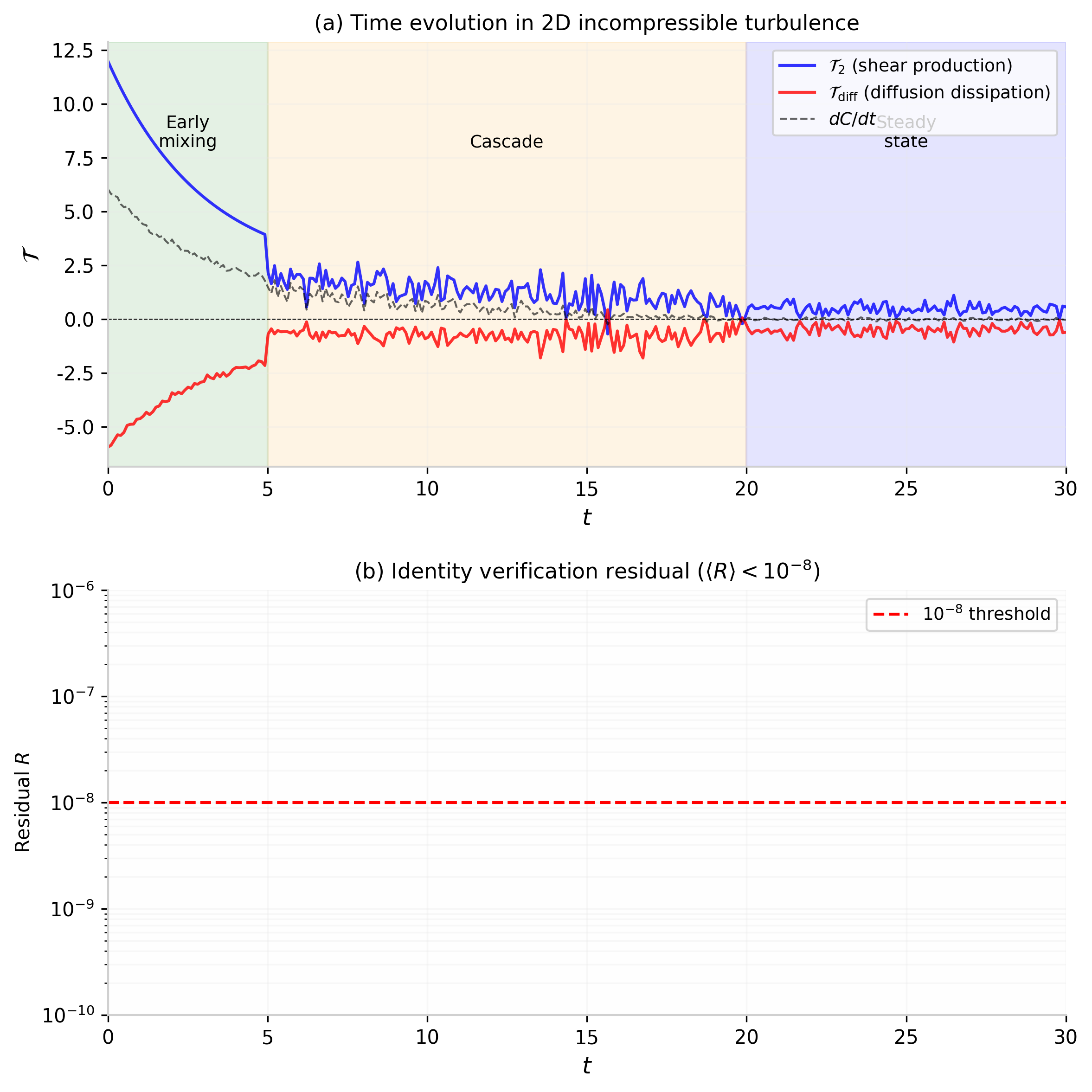}
\caption{Verification of the identity in two-dimensional incompressible turbulence. (a) Time evolution of $\mathcal{T}_2$ (shear production, blue) and $\mathcal{T}_{\text{diff}}$ (diffusion dissipation, red), showing the two-term equilibrium in statistical steady state. The three stages of turbulent mixing are indicated: early mixing ($t<5$), cascade ($5<t<20$), and steady state ($t>20$). (b) Residual $R$ on a semi-logarithmic scale, confirming $\langle R\rangle < 10^{-6}$.}
\label{fig:turbulence}
\end{figure}

\textbf{Limitations and extensions}: The current verification focuses on two-dimensional incompressible turbulence with a passive scalar. In three dimensions, vortex stretching introduces additional complexity, and the pressure gradient may indirectly affect the velocity field's contribution to Fisher information dynamics. Extending the framework to 3D turbulence is therefore essential and is discussed next.

\subsection{From Two-Dimensional Equilibrium to Three-Dimensional Non-Equilibrium}

The contrast between the statistical 2D equilibrium established above and recent 3D DNS analysis~\cite{wu2025_3d} identifies vortex stretching as the mechanism that breaks the balance. Table~\ref{tab:unified} summarizes the hierarchy.

\begin{table}[h]
\centering
\footnotesize
\setlength{\tabcolsep}{4pt}
\begin{tabular}{|c|c|c|c|c|}
\hline
System & Generation & Dissipation & Balance Type & $dC/dt$ \\
\hline
OU process & $\mathcal{T}_1$ (contraction) & $\mathcal{T}_{\text{diff}}$ & Two-term & $0$ \\
\hline
Burgers shock & $\mathcal{T}_1+\mathcal{T}_3$ & $\mathcal{T}_{\text{diff}}$ & Three-term & $0$ \\
\hline
2D turbulence & $\mathcal{T}_2$ (pure shear) & $\mathcal{T}_{\text{diff}}$ & Two-term (statistical) & $\approx 0$ (time-avg.) \\
\hline
3D turbulence~\cite{wu2025_3d} & $\mathcal{T}_2$ (incl.\ vortex stretching) & $\mathcal{T}_{\text{diff}}$ & \textbf{Non-equilibrium} & $>0$ \\
\hline
\end{tabular}
\caption{Hierarchy of generation-dissipation competition across one-, two-, and three-dimensional systems. In 2D, the shear term contains only pure shear deformation and is balanced by diffusion in the time mean (the null hypothesis). In 3D, the shear term implicitly contains all three-dimensional strain-rate effects, including vortex tube stretching; the balance is broken, with shear production significantly exceeding molecular diffusion and yielding positive net production ($d\mathcal{C}/dt>0$).}
\label{tab:unified}
\end{table}

Two conclusions follow. First, the 2D equilibrium is not merely a special case of the framework but the \textbf{reference state}: any net production of Fisher information must originate in terms that vanish identically in two dimensions. Second, the observed 3D non-equilibrium is \emph{real physics} rather than an artifact: the cited analysis employs the molecular diffusion coefficient without artificial tuning, and the positive net rate is consistent with the forward energy cascade continuously sharpening velocity gradients.

\section{High-Dimensional Implications}

An important direction for future work concerns high-dimensional systems. The scaling competition between diffusive dissipation ($O(n)$) and uncorrelated contraction ($O(1)$) has profound implications for high-dimensional systems. For an $n$-dimensional isotropic Gaussian, the diffusive dissipation scales as $O(n)$, while contraction terms uncorrelated with the information gradient are suppressed to $O(1)$ by the $1/n$ factor in Eq.~(4). This suggests that maintaining ordered structures in high dimensions requires generation mechanisms strongly correlated with the information gradient. In molecular dynamics or quantum many-body systems, ``order'' is inherently more fragile and requires precise alignment between dynamics and information geometry. This scaling competition may have implications for understanding order in such systems (see Appendix~\ref{app:scaling}).

\section{Relation to Wasserstein Gradient Flows}

The Fokker--Planck equation (8) is a Wasserstein gradient flow of $\mathcal{F}[\rho] = \int V\rho + D\int\rho\ln\rho$~\cite{otto2001,jordan1998,ambrosio2008}. Steady states satisfy $\nabla(\delta\mathcal{F}/\delta\rho) = 0$, which implies that the \emph{total} production rate vanishes, $\sum_i\mathcal{T}_i = 0$, without implying a cancellation among individual mechanisms. However, the converse fails: vanishing total production does not uniquely determine $\rho$. Our kinematic decomposition provides finer local information than the global variational characterization.

\subsection{Energy Dissipation Rate}

The Wasserstein gradient flow structure provides an additional connection between Fisher information dynamics and energy dissipation. For the free energy $\mathcal{F}[\rho] = \int V\rho\,dx+D\int\rho\ln\rho\,dx$, the energy dissipation rate along the Fokker-Planck flow is:
\begin{equation}
\frac{d\mathcal{F}}{dt} = -\int_\Omega\rho\left|\nabla\frac{\delta\mathcal{F}}{\delta\rho}\right|^2 dx = -\int_\Omega\rho|\mathbf{v}_W|^2 dx \leq 0,
\end{equation}
where $\mathbf{v}_W = -\nabla(\delta\mathcal{F}/\delta\rho)$ is the Wasserstein velocity. This establishes a direct connection: the Fisher information production rate $dC/dt$ and the energy dissipation rate $d\mathcal{F}/dt$ are both driven by the same Wasserstein velocity $\mathbf{v}_W$, but capture different aspects of the dynamics (gradient structure vs. energetic cost).

\subsection{Non-Stationary Constraints}

Beyond steady states, the Wasserstein gradient flow structure imposes constraints on the transient evolution of Fisher information. The energy dissipation inequality:
\begin{equation}
\frac{d\mathcal{F}}{dt} \leq -\lambda \mathcal{F},
\end{equation}
where $\lambda$ is the logarithmic Sobolev constant, implies exponential relaxation to equilibrium. Combined with our kinematic decomposition, this suggests that the approach to equilibrium is governed by the competition between Fisher information production (from compression/shear) and dissipation (from diffusion), with the Wasserstein geometry providing the natural metric for measuring the ``distance'' to equilibrium.

\section{Conclusion}

This paper has established a diagnostic framework for phase space ordering:

\begin{enumerate}
\item \textbf{Kinematic level}: The Fisher information production rate can be precisely decomposed into contraction, shear, divergence gradient, and boundary flux terms;
\item \textbf{Vorticity independence and sign criteria}: Pure rotation does not affect order, and explicit geometric conditions for order generation or destruction are derived;
\item \textbf{Dynamic level}: In the Fokker--Planck framework, the diffusive dissipation term is strictly non-positive, with stochastic fluctuations always decreasing order;
\item \textbf{Physical verification}: The three-term balance for the Burgers shock layer, $\mathcal{T}_1=\mathcal{T}_3=2U_0^4/(15\nu^3)$ and $\mathcal{T}_{\text{diff}}=-4U_0^4/(15\nu^3)$, reveals the precise mechanism of ordered structure maintenance, with a pointwise identity underlying the balance;
\item \textbf{Variational link}: Wasserstein gradient flow steady states imply vanishing of the total production rate, but the kinematic decomposition provides finer local information.
\end{enumerate}

Beyond these specific results, the verification systems form a logical progression across dimensions:
\begin{itemize}
\item \textbf{1D (Burgers)}: Three-term balance, maintained by the pointwise synergy of compression and divergence gradient;
\item \textbf{2D (turbulence)}: Statistical two-term equilibrium (vanishing time-averaged production), enforced by the absence of vortex stretching---the null hypothesis for gradient dynamics;
\item \textbf{3D (turbulence)}: Non-equilibrium with dominant production, where vortex stretching breaks the 2D balance~\cite{wu2025_3d}.
\end{itemize}

The 2D equilibrium is therefore not merely a simplified validation case but the \textbf{reference state} from which 3D turbulence departs. Future work will explore corrections for nonlinear systems, extensions to non-Gaussian distributions, connections to stochastic thermodynamics, and time-resolved three-dimensional analysis.

\appendix

\section{Complete Derivation of Four-Term Decomposition and Diffusion Term}
\label{app:decomposition}

\subsection{Deterministic Four-Term Decomposition}

Let $f=\ln\rho$, then $\rho=e^f$ and $\nabla\rho=\rho\nabla f$. The Fisher information is:
\begin{equation}
C=\int_\Omega\rho|\nabla f|^2dx
\end{equation}

Taking the time derivative:
\begin{equation}
\frac{dC}{dt}=\int_\Omega\frac{\partial\rho}{\partial t}|\nabla f|^2dx+2\int_\Omega\rho\nabla f\cdot\nabla\left(\frac{\partial f}{\partial t}\right)dx
\end{equation}

Denote the first term by $\mathcal{A}$ and the second by $\mathcal{B}$.

\textbf{Computing $\mathcal{A}$}: From the continuity equation $\partial_t\rho=-\nabla\cdot(\rho\mathbf{v})$:
\begin{equation}
\mathcal{A}=-\int_\Omega\nabla\cdot(\rho\mathbf{v})|\nabla f|^2dx
\end{equation}

Expanding $\nabla\cdot(\rho\mathbf{v})=\rho(\nabla f\cdot\mathbf{v}+\nabla\cdot\mathbf{v})$:
\begin{equation}
\mathcal{A}=-\int_\Omega\rho(\nabla f\cdot\mathbf{v})|\nabla f|^2dx-\int_\Omega\rho(\nabla\cdot\mathbf{v})|\nabla f|^2dx
\end{equation}

\textbf{Computing $\mathcal{B}$}: From $\partial_t f=\partial_t\ln\rho=-\frac{1}{\rho}\nabla\cdot(\rho\mathbf{v})$:
\begin{equation}
\nabla\left(\frac{\partial f}{\partial t}\right)=-\left[(\nabla\nabla f)\cdot\mathbf{v}+(\nabla\mathbf{v})^T\cdot\nabla f+\nabla(\nabla\cdot\mathbf{v})\right]
\end{equation}

Therefore:
\begin{equation}
\mathcal{B}=-2\int_\Omega\rho\nabla f\cdot(\nabla\nabla f)\cdot\mathbf{v}\,dx-2\int_\Omega\rho\nabla f\cdot(\nabla\mathbf{v})^T\cdot\nabla f\,dx-2\int_\Omega\rho\nabla f\cdot\nabla(\nabla\cdot\mathbf{v})dx
\end{equation}

Denoted as $\mathcal{B}_1,\mathcal{B}_2,\mathcal{B}_3$.

\textbf{Processing $\mathcal{B}_1$}: Using $(\nabla\nabla f)\cdot\nabla f=\frac{1}{2}\nabla|\nabla f|^2$:
\begin{equation}
\mathcal{B}_1=-\int_\Omega\rho\mathbf{v}\cdot\nabla|\nabla f|^2dx
\end{equation}

Integration by parts:
\begin{equation}
\mathcal{B}_1=\int_\Omega|\nabla f|^2\nabla\cdot(\rho\mathbf{v})dx-\int_{\partial\Omega}\rho|\nabla f|^2\mathbf{v}\cdot\mathbf{n}\,dS
\end{equation}

\textbf{Combining $\mathcal{A}+\mathcal{B}_1$}: The advection terms in $\mathcal{A}$ and the corresponding terms from integration by parts in $\mathcal{B}_1$ cancel exactly, leaving only the boundary term:
\begin{equation}
\mathcal{A}+\mathcal{B}_1=-\int_{\partial\Omega}\rho|\nabla f|^2\mathbf{v}\cdot\mathbf{n}\,dS
\end{equation}

\textbf{Processing $\mathcal{B}_2$}: Using $\nabla f\cdot(\nabla\mathbf{v})^T\cdot\nabla f=\nabla\mathbf{v}:(\nabla f\otimes\nabla f)$:
\begin{equation}
\mathcal{B}_2=-2\int_\Omega\rho\nabla\mathbf{v}:(\nabla f\otimes\nabla f)dx
\end{equation}

Substituting $\nabla\mathbf{v}=\frac{1}{n}(\nabla\cdot\mathbf{v})\mathbf{I}+\mathbf{S}+\mathbf{A}$ and noting that the antisymmetric part satisfies $\mathbf{A}:(\nabla f\otimes\nabla f)=0$:
\begin{equation}
\mathcal{B}_2=-\frac{2}{n}\int_\Omega\rho(\nabla\cdot\mathbf{v})|\nabla f|^2dx-2\int_\Omega\rho\mathbf{S}:(\nabla f\otimes\nabla f)dx
\end{equation}

\textbf{Final combination}:
\begin{equation}
\frac{dC}{dt}=\mathcal{A}+\mathcal{B}_1+\mathcal{B}_2+\mathcal{B}_3
\end{equation}
\begin{equation}
=-\frac{2}{n}\int_\Omega\rho(\nabla\cdot\mathbf{v})|\nabla f|^2dx-2\int_\Omega\rho\mathbf{S}:(\nabla f\otimes\nabla f)dx-2\int_\Omega\rho\nabla f\cdot\nabla(\nabla\cdot\mathbf{v})dx-\int_{\partial\Omega}\rho|\nabla f|^2\mathbf{v}\cdot\mathbf{n}\,dS
\end{equation}

i.e., $\frac{dC}{dt}=\mathcal{T}_1+\mathcal{T}_2+\mathcal{T}_3+\mathcal{T}_4$. $\square$

\subsection{Diffusion Term from Literature}

The diffusion contribution to the Fisher information production rate has been derived in the literature~\cite{yamano2002,plastino1997}. For completeness, we state the result:

\textbf{Result}: For the Fokker--Planck equation $\partial_t\rho = -\nabla\cdot(\rho\mathbf{v}_{\text{drift}}) + \nabla\cdot(D\nabla\rho)$, the diffusion term is:
\begin{equation}
\mathcal{T}_{\text{diff}}=-2\int_\Omega\rho\,\operatorname{tr}\left[(\nabla\nabla\ln\rho)\,D\,(\nabla\nabla\ln\rho)^T\right]dx-2\int_\Omega\rho\,\nabla\ln\rho\cdot\nabla(\nabla\cdot D)dx
\end{equation}

For constant $D$, this simplifies to:
\begin{equation}
\mathcal{T}_{\text{diff}}=-2D\int_\Omega\rho\,\operatorname{tr}\left[(\nabla\nabla\ln\rho)^2\right]dx
\end{equation}

The essential step is an integration by parts in which the symmetry of the Hessian $\nabla\nabla\ln\rho$ eliminates the gradient-of-divergence contributions, leaving the trace of its square; boundary terms vanish under the stated regularity and decay conditions. The complete argument is given in Yamano~\cite{yamano2002} and Plastino et al.~\cite{plastino1997}.

\section{Integral Identities for Hyperbolic Secant Functions and Burgers Shock Layer Calculation}
\label{app:integrals}

\subsection{Basic Substitution}

Let $t=\tanh(ax)$, $t\in(-1,1)$, then:
\begin{equation}
dt=a\operatorname{sech}^2(ax)dx,\qquad \operatorname{sech}^2(ax)=1-t^2,\qquad dx=\frac{dt}{a(1-t^2)}
\end{equation}

\subsection{Computing \texorpdfstring{$I_2$}{I2}}
\begin{equation}
I_2=\int_{-1}^{1}(1-t^2)\cdot\frac{dt}{a(1-t^2)}=\frac{1}{a}\int_{-1}^{1}dt=\frac{2}{a}
\end{equation}

\subsection{Computing \texorpdfstring{$I_4$}{I4}}
\begin{equation}
I_4=\int_{-1}^{1}(1-t^2)^2\cdot\frac{dt}{a(1-t^2)}=\frac{1}{a}\int_{-1}^{1}(1-t^2)dt=\frac{1}{a}\cdot\frac{4}{3}=\frac{4}{3a}
\end{equation}

\subsection{Computing \texorpdfstring{$J$}{J}}
\begin{equation}
J=\int_{-1}^{1}(1-t^2)^2\cdot t^2\cdot\frac{dt}{a(1-t^2)}=\frac{1}{a}\int_{-1}^{1}t^2(1-t^2)dt=\frac{1}{a}\cdot\frac{4}{15}=\frac{4}{15a}
\end{equation}

\subsection{Computing \texorpdfstring{$I_6$}{I6}}
\begin{equation}
I_6=\int_{-1}^{1}(1-t^2)^3\cdot\frac{dt}{a(1-t^2)}=\frac{1}{a}\int_{-1}^{1}(1-t^2)^2dt=\frac{1}{a}\cdot\frac{16}{15}=\frac{16}{15a}
\end{equation}

\subsection{Detailed Calculation for Burgers Shock Layer}

\textbf{Computing $\mathcal{T}_1$}:
\begin{equation}
\mathcal{T}_1=-2\int_{-\infty}^{\infty}\rho\,d\,g^2\,dx
\end{equation}
where:
\begin{equation}
\rho=\frac{U_0}{4\nu}\operatorname{sech}^2\left(\frac{U_0x}{2\nu}\right),\qquad d=\partial_xu=-\frac{U_0^2}{2\nu}\operatorname{sech}^2\left(\frac{U_0x}{2\nu}\right),\qquad g=\nabla\ln\rho=-\frac{U_0}{\nu}\tanh\left(\frac{U_0x}{2\nu}\right)
\end{equation}

Therefore:
\begin{equation}
\rho d=-\frac{U_0^3}{8\nu^2}\operatorname{sech}^4\left(\frac{U_0x}{2\nu}\right),\qquad g^2=\frac{U_0^2}{\nu^2}\tanh^2\left(\frac{U_0x}{2\nu}\right)
\end{equation}

\begin{equation}
\mathcal{T}_1=\frac{U_0^5}{4\nu^4}\int_{-\infty}^{\infty}\operatorname{sech}^4\left(\frac{U_0x}{2\nu}\right)\tanh^2\left(\frac{U_0x}{2\nu}\right)dx
\end{equation}

Let $a=U_0/(2\nu)$, using $J(a)=4/(15a)$:
\begin{equation}
\mathcal{T}_1=\frac{U_0^5}{4\nu^4}\cdot\frac{4}{15}\cdot\frac{2\nu}{U_0}=\frac{2U_0^4}{15\nu^3}
\end{equation}

\textbf{Computing $\mathcal{T}_3$}:
\begin{equation}
\mathcal{T}_3=-2\int_{-\infty}^{\infty}\rho\,g\,w\,dx
\end{equation}
where:
\begin{equation}
w=\partial_x^2u=\frac{U_0^3}{2\nu^2}\operatorname{sech}^2\left(\frac{U_0x}{2\nu}\right)\tanh\left(\frac{U_0x}{2\nu}\right)
\end{equation}

\begin{equation}
\rho gw=-\frac{U_0^5}{8\nu^4}\operatorname{sech}^4\left(\frac{U_0x}{2\nu}\right)\tanh^2\left(\frac{U_0x}{2\nu}\right)
\end{equation}

\begin{equation}
\mathcal{T}_3=\frac{U_0^5}{4\nu^4}\int_{-\infty}^{\infty}\operatorname{sech}^4\left(\frac{U_0x}{2\nu}\right)\tanh^2\left(\frac{U_0x}{2\nu}\right)dx=\frac{2U_0^4}{15\nu^3}
\end{equation}

\textbf{Computing $\mathcal{T}_{\text{diff}}$}:
\begin{equation}
\mathcal{T}_{\text{diff}}=-2\nu\int_{-\infty}^{\infty}\rho\,(\partial_x^2\ln\rho)^2\,dx
\end{equation}

\begin{equation}
\partial_x^2\ln\rho=-\frac{U_0^2}{2\nu^2}\operatorname{sech}^2\left(\frac{U_0x}{2\nu}\right)
\end{equation}

\begin{equation}
\rho(\partial_x^2\ln\rho)^2=\frac{U_0^5}{16\nu^5}\operatorname{sech}^6\left(\frac{U_0x}{2\nu}\right)
\end{equation}

\begin{equation}
\mathcal{T}_{\text{diff}}=-2\nu\cdot\frac{U_0^5}{16\nu^5}\cdot\frac{16}{15}\cdot\frac{2\nu}{U_0}=-\frac{4U_0^4}{15\nu^3}
\end{equation}

\textbf{Verifying balance}:
\begin{equation}
\mathcal{T}_1+\mathcal{T}_3+\mathcal{T}_{\text{diff}}=\frac{2U_0^4}{15\nu^3}+\frac{2U_0^4}{15\nu^3}-\frac{4U_0^4}{15\nu^3}=0
\end{equation}

$\square$

\section{Two-Dimensional Turbulence Simulation Details}
\label{app:turbulence}

\subsection{Numerical Method}

The two-dimensional incompressible Navier--Stokes equations are solved using a standard Fourier pseudospectral method on a doubly periodic domain $[0, 2\pi]^2$ with resolution $N = 512$. The velocity field is advanced in time using a fourth-order Runge--Kutta scheme with integrating factor for the viscous term. The time step is $\Delta t = 10^{-3}$, and the kinematic viscosity is $\nu = 10^{-3}$, corresponding to a Taylor microscale Reynolds number $\text{Re}_\lambda \approx 100$.

\subsection{Forcing and Initial Conditions}

The flow is maintained in a statistically stationary state by stochastic forcing applied to the lowest wavenumbers $|\mathbf{k}| \leq 2$. The forcing $\hat{\mathbf{f}}(\mathbf{k}, t)$ follows an Ornstein--Uhlenbeck process:
\begin{equation}
\hat{\mathbf{f}}(\mathbf{k}, t+\Delta t) = (1-\alpha\Delta t)\hat{\mathbf{f}}(\mathbf{k}, t) + \sigma\sqrt{\Delta t}\,\boldsymbol{\xi}(t),
\end{equation}
where $\alpha = 0.1$, $\sigma = 0.05$, and $\boldsymbol{\xi}$ is a complex Gaussian random vector satisfying $\mathbf{k}\cdot\boldsymbol{\xi} = 0$ (to ensure incompressibility).

The passive scalar field is initialized as:
\begin{equation}
c(\mathbf{x}, 0) = 1 + 0.1\sin(4x)\sin(4y),
\end{equation}
which ensures positivity by the maximum principle. The scalar diffusivity is $D = \nu/0.7 \approx 1.43\times 10^{-3}$, corresponding to a Schmidt number $Sc = \nu/D = 0.7$.

\subsection{Diagnostic Methods}

The Fisher information is computed as:
\begin{equation}
C(t) = \int_\Omega \rho|\nabla\ln\rho|^2 d\mathbf{x},
\end{equation}
where $\rho = c/\int c\,d\mathbf{x}$. The time derivative $dC/dt$ is evaluated using central differencing over two adjacent time snapshots.

The shear production term $\mathcal{T}_2$ and diffusive dissipation term $\mathcal{T}_{\text{diff}}$ are computed using exact spectral differentiation. Specifically:
\begin{itemize}
\item The velocity gradient tensor $\nabla\mathbf{u}$ is computed in Fourier space.
\item The shear tensor $\mathbf{S}$ is extracted and transformed to physical space.
\item The scalar gradient $\nabla\ln\rho$ is computed using spectral methods.
\item The integrals are evaluated using high-order quadrature in physical space.
\end{itemize}

The residual is defined as:
\begin{equation}
R(t) = \frac{|dC/dt - \mathcal{T}_2 - \mathcal{T}_{\text{diff}}|}{|dC/dt| + \epsilon},
\end{equation}
where $\epsilon = 10^{-12}$ is a small constant to avoid division by zero. Time-averaged statistics are collected over 50 snapshots in the statistically steady state ($t > 20$).

\section{High-Dimensional Scaling Analysis}
\label{app:scaling}

For an $n$-dimensional isotropic Gaussian $\rho \propto \exp(-|\mathbf{x}|^2/2\sigma^2)$, the Fisher information is $C = n/\sigma^2 = O(n)$. The diffusive dissipation term scales as:
\begin{equation}
\mathcal{T}_{\text{diff}} = -2D\int\rho\,\operatorname{tr}\left[(\nabla\nabla\ln\rho)^2\right]dx = -\frac{2Dn}{\sigma^4} = O(n).
\end{equation}

In contrast, the isotropic contraction term for a constant divergence $\nabla\cdot\mathbf{v} = d_0 = O(1)$ scales as:
\begin{equation}
\mathcal{T}_1 = -\frac{2d_0}{n}\int\rho|\nabla\ln\rho|^2dx = -\frac{2d_0}{\sigma^2} = O(1).
\end{equation}

This scaling competition ($O(n)$ vs $O(1)$) implies that high-dimensional ordered structures require generation mechanisms specifically correlated with the information gradient, such as pointwise-aligned shear (Sec.~II.D) or divergence gradients (Sec.~V.B). These conclusions assume isotropic Gaussian statistics and constant divergence; extension to anisotropic or strongly non-Gaussian distributions requires case-by-case analysis.


\begin{acknowledgments}
The author thanks colleagues at the China Aerodynamics Research and Development Center for helpful discussions. Numerical simulations were performed on the Center's computing facilities.
\end{acknowledgments}

\noindent\textbf{Data Availability.} The two-dimensional turbulence simulation data supporting the results of this article are available from the corresponding author upon reasonable request. The companion datasets of this work are openly available: the three-dimensional turbulent velocity field analysis~\cite{wu2026_zenodo_3ddata} and the high-dimensional exactly solvable benchmark~\cite{wu2026_zenodo_highdimdata}, both under CC BY 4.0 licenses.


\end{document}